\documentclass{acmtrans2m}

\input{epsf}

\usepackage[ruled,vlined,linesnumbered,commentsnumbered]{algorithm2e}
\usepackage{amsmath,amssymb,amsfonts}
\usepackage{enumitem}

\begin{document}

\newtheorem{theorem}    {Theorem} [section]
\newtheorem{lemma}      [theorem] {Lemma}
\newtheorem{corollary}  [theorem] {Corollary}
\newtheorem{claim}      [theorem] {Claim}
\newtheorem{definition} [theorem] {Definition}
\newtheorem{invariant}  [theorem] {Invariant}
\newtheorem{assumption} [theorem] {Assumption}

\markboth{Guy Kortsarz and Zeev Nutov}{A $1.5$-approximation for augmenting edge-connectivity from $1$ to $2$}

\title{A simplified $1.5$-approximation algorithm for augmenting edge-connectivity of a graph from 1 to 2}
       
\author{
Guy Kortsarz\thanks{Partially supported by NSF grant number 1218620.} \\
\small Rutgers University, Camden \\
\small {\tt guyk@crab.rutgers.edu}
\and Zeev Nutov \\
\small The Open University of Israel \\
\small {\tt nutov@openu.ac.il}
       }

\begin{abstract} 
The {\sf Tree Augmentation Problem} ({\sf TAP}) is: given a connected graph $G=(V,{\cal E})$ and an edge
set $E$ on $V$ find a minimum size subset of edges 
$F \subseteq E$ such that $(V,{\cal E} \cup F)$ is $2$-edge-connected.
In the conference version \cite{EFKN-APPROX} was sketched a $1.5$-approximation algorithm for the problem.  
Since a full proof was very complex and long, the journal version was cut into two parts.
In the first part \cite{EFKN-TALG} was only proved ratio $1.8$. 
An attempt to simplify the second part produced an error in \cite{EKN-IPL}.
Here we give a correct, different, and self contained proof of the ratio $1.5$, 
that is also substantially simpler and shorter than the previous proofs.
\end{abstract}

            

\begin{bottomstuff} 
\end{bottomstuff}

\maketitle

%%%%%%%%%%%%%%%%%%%%%%%%%%%%%%%%%%%%%%
\section{Introduction} \label{s:intro}
%%%%%%%%%%%%%%%%%%%%%%%%%%%%%%%%%%%%%%

%%%%%%%%%%%
\subsection{Problem definition and our result} \label{ss:problem-result}
%%%%%%%%%%%

% \textcolor{red}{text}
A graph (possibly with parallel edges) is {\em $k$-edge-connected} if there are $k$ pairwise
edge-disjoint paths between every pair of its nodes. 
We study the following fundamental problem: 
given a connected undirected graph $G=(V,{\cal E})$ and a
set of additional edges (called ``links") $E$ on $V$ disjoint to ${\cal E}$,
find a minimum size edge set $F \subseteq E$ so that $G \cup F=(V,{\cal E} \cup F)$ 
is $2$-edge-connected.
The $2$-edge-connected components of the given graph $G$ form a tree.
It follows that by contracting these components, one may assume that
$G$ is a tree.  Hence, our problem is:

\begin{center}\fbox{\begin{minipage}{0.965\textwidth}
\noindent
{\sf Tree Augmentation Problem (\sf TAP)} \\
{\em Instance:}  \ A tree $T=(V,{\cal E})$ and a set of links $E$ on $V$ disjoint to ${\cal E}$. \\
{\em Objective:} Find a minimum size subset $F \subseteq E$ of links such that $T \cup F$ is
$2$-edge-connected.
\end{minipage}}\end{center}

{\sf TAP} is sometimes posed as the problem of covering a laminar family. 
Namely, given a laminar family ${\cal E}$ on a groundset $V$,
and an edge set $E$ on $V$, find a minimum size $F \subseteq E$
such that for every $S \in {\cal E}$, there is an edge in $F$ with one
endpoint in $S$ and the other in $V \setminus S$. 
{\sf TAP} is also equivalent to the problem of augmenting the edge-connectivity 
from $k$ to $k+1$ for any odd $k$;
this is since the family of minimum cuts of a $k$-connected graph with $k$ odd is laminar.

The first $2$-approximation for {\sf TAP} was given by \cite{FJ}, 
where it was also shown to be APX-hard. 
Achieving ratio below $2$ was posed by 
\cite{K} as one of the main open problems in connectivity augmentation.
\cite{N} presented a $(1.875+\varepsilon)$-approximation scheme for {\sf TAP}, but his 
analysis over 30 pages is long and complex.
In the conference version \cite{EFKN-APPROX} was sketched a $1.5$-approximation algorithm for the problem.  
Since a full proof was very complex and long (40 pages), the journal version was cut into two parts.
In the first part \cite{EFKN-TALG} was only proved ratio $1.8$. 
An attempt to simplify the second part produced an error in \cite{EKN-IPL}.
Here we give a correct, different, and self contained proof of the ratio $1.5$, 
that is also substantially simpler and shorter than the previous proofs.

\begin{theorem} \label{t:main}
{\sf TAP} admits a $1.5$-approximation algorithm.
\end{theorem}

%%%%%%%%%%%
\subsection{Related work} \label{ss:related-work}
%%%%%%%%%%%

In the more general {\sf Weighted TAP} problem, the links in $E$ have weights and the
goal is to find a minimum weight augmenting edge set $F$ such that $T \cup F$ is $2$-edge connected.
There are several $2$-approxima\-ti\-on algorithms for this problem. 
The first algorithm, by \cite{FJ} was simplified later by \cite{KT}.  
These algorithms compute a minimum weight arborescence in a related directed graph.  
The primal-dual algorithm of \cite{GGPS} is another combinatorial $2$-approximation algorithm for the problem.
The iterative rounding algorithm of \cite{Jain} is an LP-based $2$-approximation algorithms. 
The approximation ratio of $2$ for all these algorithms is tight even for {\sf TAP}.
Breaking the ratio of $2$ for {\sf Weighted TAP} is a major open problem in approximation theory.
In \cite{CN} is given an algorithm that computes a 
$(1 + \ln 2)$-approximate solution for constant diameter trees.

{\sf TAP} is APX-hard even if the set $E$ of links forms a cycle on the leaves of $T$ \cite{CJR}.
A natural cut-LP for {\sf TAP} has integrality gap at least $3/2$ \cite{CKKK}.
\cite{MN} gave ratio $17/12$ for the special case of {\sf TAP} when every link connects two leaves, 
and obtained ratio $3/2$ for this version w.r.t. a leaf edge-cover LP.
\cite{KN-LP} showed that a slightly modified LP has integrality gap $1.75$ for {\sf TAP}.
Studying various LP-relaxations for {\sf TAP} is motivated by the hope that these may lead to 
breaking the ratio of $2$ for {\sf Weighted TAP}.

%%%%%%%%%%%
\subsection{Organization} \label{ss:organization}
%%%%%%%%%%%

In the next Section~\ref{s:preliminaries} we define some special types of trees and show some 
properties of these trees. These are needed to state our lower bound given in Section~\ref{s:lb}.
In Section~\ref{s:credit-algo} we explain how we use our lower bound and describe the algorithm,
relying on a certain lemma; this lemma is proved in Sections \ref{s:B'} and \ref{s:char}.

%%%%%%%%%%%%%%%%%%%%%%%%%%%%%%%%%%%%%%%%%%%%%%%%%%%%%%%%%%%%%%%%%%%%%%%%%%%%%%%%%%%%%%%%%%%%
\section{Preliminaries: some small trees and shadows-minimal covers} \label{s:preliminaries}
%%%%%%%%%%%%%%%%%%%%%%%%%%%%%%%%%%%%%%%%%%%%%%%%%%%%%%%%%%%%%%%%%%%%%%%%%%%%%%%%%%%%%%%%%%%%

Let $T=(V,{\cal E})$ be a tree. For $u,v \in V$ let $(u,v)\in {\cal E}$ denote the edge in $T$ and 
$uv$ the link in $E$ between $u$ and $v$. 
Let $P(uv)=P_T(uv)$ denote the path between $u$ and $v$ in $T$.
A link $uv$ {\em covers} all the edges along the path $P(uv)$.
We designate a node $r$ of $T$ as the {\em root}, and 
refer to the pair $T,r$ as a {\em rooted tree}
(we do not mention the root when it is clear from the context).
The choice of $r$ defines a partial order on $V$: $u$ is a {\em descendant} of $v$ 
and $v$ is an {\em ancestor} of $u$ if $v$ belongs to $P(ru)$;
if, in addition, $(u,v) \in T$, then $u$ is a {\em child} of
$v$, and $v$ is the {\em parent} of $u$.  
The {\em leaves} of $T$ are the nodes in $V \setminus \{r\}$ that have no descendants. 
We denote the leaf set of $T$ by $L(T)$, or simply by $L$, when the context is clear.  
The {\em rooted subtree} of $T$ induced by $v$ and its descendants is
denoted by $T_v$ ($v$ is the root of $T_v$). A subtree $T'$
of $T$ is called a {\em rooted subtree} of $T$ if $T'=T_v$ for some $v \in V$. 

\begin{definition} [(shadow)] \label{d:shadow}
A link $u'v'$ is a {\em shadow} of a link $uv$ if $P(u'v') \subseteq P(uv)$.
An inclusion minimal cover $F$ of $T$ is {\em shadows-minimal} if for
every link $uv \in F$ replacing $uv$ by any proper shadow of
$uv$ results in a set of links that does not cover $T$.
\end{definition}

Every {\sf TAP} instance can be rendered closed under shadows by adding all shadows of existing links.  
We refer to the addition of all shadows as {\em shadow completion}.  
Shadow completion does not affect the optimal solution size, since every shadow can be replaced by some
link covering all edges covered by the shadow. Thus we may assume the following.

\begin{assumption} \label{ass:sh-completion}
The set of links $E$ is closed under shadows, that is, if $uv \in E$
and $P(u'v') \subseteq P(uv)$ then $u'v' \in E$.
\end{assumption}

The {\em up-link} ${up}(a)$ of a node $a$ is the link $au$ such that $u$ is as close 
as possible to the root; such $u$ is called the {\em up-node} of $a$.
Under Assumption~\ref{ass:sh-completion}, $u$ is unique and is an ancestor of $a$.
For a rooted subtree $T'$ of $T$ and a node $a \in T'$, we say that 
$T'$ is {\em $a$-closed} if the up-node of $a$ belongs to $T'$ 
(namely, if no link incident to $a$ has its other endnode outside $T'$), 
and $T'$ is {\em $a$-open} otherwise.

\begin{definition} [(twin link, stem)] \label{d:twin}
A link between leaves $a,b$ of $T$ is a {\em twin link} if its contraction results in a new leaf; 
$a,b$ are called {\em twins} and their least common ancestor $s$ is called a {\em stem} (see Fig.~\ref{f:twin-lock}(a)).
\end{definition}

\begin{definition} [(locked leaf, locking link, locking tree)] \label{d:lock} 
A leaf $a$ of $T$ is {\em locked} by a link $bb'$, and $bb'$ is a {\em locking link} of $a$ 
if there exists a rooted proper subtree $T_v$ of $T$ such that 
$L(T_v)=\{a,b,b'\}$, 
$ab$ is a twin link, 
and $T_v$ is $a$-closed;
such minimal $T_v$ is called the {\em locking tree} of $a$ (see Fig.~\ref{f:twin-lock}(b)). 
\end{definition}

\begin{figure}
\centering 
\epsfbox{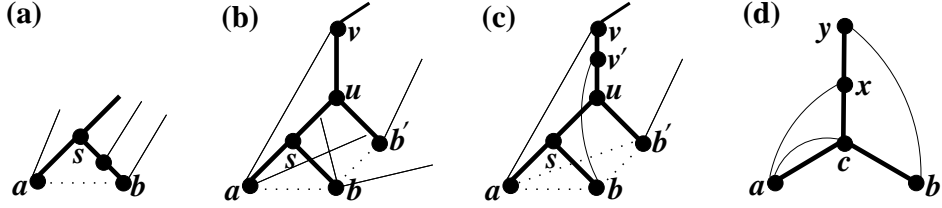}
\caption{(a,b,c) Illustration to Definitions \ref{d:twin} and \ref{d:lock}.
         The twin-link and the locking link are shown by dotted lines,
         some other possible links in $E$ are shown by solid thin lines.
         Some tree edges may be paths (the $uv$-path in (b,c) may have length zero). 
         (d) A link $by$ that overlaps a link $ax$.}
\label{f:twin-lock}
\end{figure}

Note that if $ab$ is a twin link and $a$ is locked by $bb'$,
then $b$ may be locked by $ab'$ (see Fig.~\ref{f:twin-lock}(c)). 
In this case the locking tree of one of $a,b$ contains the other, and whenever 
we will use the notation as above, we will assume w.l.o.g. that the locking tree
of $a$ contains the locking tree of $b$ (see trees $T_v$ and $T_{v'}$ in Fig.~\ref{f:twin-lock}(c)).

For $X,Y \subseteq V$ and a link set $F$, 
let $F(X,Y)=\{xy \in F: x \in X,y \in Y\}$ denote
the set of links in $F$ that have one endpoint in $X$ and the
other in $Y$; for $x \in V$ let $d_F(x)=|F(x,V)|$ be the degree of $x$ w.r.t. $F$.
For the rest of the paper we fix $F$ to be some optimal shadows-minimal cover of $T$ 
with maximal number of twin links.
In the rest of this section we establish some properties of $F$ that we use later.

A link $by$ {\em overlaps} a link $ax$ if the paths $P(ax),P(by)$ have an edge in common 
and if one of $a,x$ belongs to $P(by)$, see Fig.~\ref{f:twin-lock}(d).

\begin{claim} \label{c:overlap}
No link in $F$ overlaps the other.
\end{claim}
\begin{proof}
Suppose to the contrary that $by \in F$ overlaps $ax \in F$.
If both $a$ and $x$ belong to $P(by)$ then $F \setminus \{ax\}$ is a cover of $T$ of size smaller than $|F|$.
Suppose that exactly one of $a,x$ belongs to $P(by)$, say $x \in P(by)$ and $a \notin P(by)$, see Fig.~\ref{f:twin-lock}(d).
Let $c$ be a node in $P(ax) \cap P(by)$ distinct from $x$
(e.g., $c$ may be the first node of $P(by)$ when we traverse $P(ax)$ from $a$ to $x$).
Then $ac$ is a proper shadow of $ax$ and $ac$ covers the edges of $P(ax)$ that are not covered by $by$.
Replacing $ax$ by its proper shadow $ac$ results in a cover of $T$, contradicting shadows-minimality of $F$.
\end{proof}

\begin{claim} \label{c:shm}
$d_F(a)=1$ for every leaf $a \in L$ of $T$.
\end{claim}
\begin{proof}
If $ax,by$ are two links incident to the same leaf $a$ of $T$, then one of them overlaps the other,
contradicting Claim~\ref{c:overlap}.  
\end{proof}

\begin{claim} \label{c:twins}
Let $a,b$ be twins with stem $s$, let $T_s=P(sa) \cup P(sb)$ be the subtree of $T$ rooted at $s$,
and let $F'$ be the set of links in $F$ with at least one endnode in $T_s$.
Then either $F'=\{ab,sz\}$ for some $z \notin T_s$, % (so $d_F(s)=1$), 
or $F'=\{ax,by\}$ for some $x,y \notin T_s$.
\end{claim}
\begin{proof}
Suppose that $ab \in F$. 
Every link that covers the parent edge of $s$ belongs to $F'$, thus $|F'| \geq 2$.
Consider any link $s'z \in F'$ with $s' \in T_s$.
Then $ab$ overlaps $s'z$, unless $s'=s$.
There cannot be another link $sz' \in F'$, since then one of $sz,s'z$ overlaps the other.
Consequently, $F$ contains a unique link $sz$, as claimed.

Suppose now that $ab \notin F$.
Let $ax$ and $by$ be the (unique, by Claim~\ref{c:shm}) link incident to $a$ and $b$, respectively.
One of $x,y$ is not in $T_s$, say $x \notin T_s$; 
otherwise $(F \setminus\{ax,by\}) \cup \{ab\}$ is a cover of $T$ of size smaller than $|F|$. 
We cannot have $y \in T_s$ since then $(F \setminus \{ax,by\}) \cup \{ab,sx\}$
is a shadows-minimal cover of $T$ of size $|F|$ with more twin links than $F$,
contradicting our choice of $F$.
There cannot be another link in $F'$ since then it will be overlapped by one of $ax,by$. 
\end{proof}

\begin{claim} \label{c:lock-ok}
Consider a locked leaf $a$ and its locking tree $T_v$ as in Definition~\ref{d:lock}, and suppose that 
$ax \in F$ for some $x \notin \{b,b'\}$. 
Then $x$ is a proper ancestor of the least common ancestor $u$ of $a,b,b'$ and $bb' \in F$.
Furthermore, there is a link $xz \in F$ such that $z \notin T_v$ and $z$ is not a locked leaf.
\end{claim}
\begin{proof}
Let $by$ and $b'y'$ be the (unique, by Claim~\ref{c:shm}) links in $F$ incident to $b$ and to $b'$, respectively.
We start by refuting the case that one of $x,y$ belongs to $T_u \setminus \{b'\}$. 
By Claim~\ref{c:twins}, $x,y \notin T_s$.
If one of $x,y$ belongs to $P(su) \setminus \{s\}$ (see Fig.~\ref{f:lock-ok}(a)),
then one of $ax,by$ overlaps the other, contradicting Claim~\ref{c:overlap}.
Suppose that $x$ belongs to $P(b'u) \setminus \{b'\}$, see Fig.~\ref{f:lock-ok}(b);
refuting the case $y \in P(b'u) \setminus \{b'\}$ is similar.
Note that any link with exactly one endnode in $P(b'x)$ overlaps $ax$.
Thus $b'y'$ has both endnodes in $P(b'x)$.
The link $e$ that covers the edge between $x$ and its child also has both endnodes in $P(b'x)$.
Thus we must have $y'=x$, by the optimality of $F$; 
otherwise $F \setminus \{b'y',e\} \cup \{b'x\}$ is a cover of $T$ of size $|F|-1$
($b'x$ is a shadow of an existing link $bb'$).
But then $F'=(F \setminus\{ax,b'x\}) \cup \{ab,b's\}$ (see Fig.~\ref{f:lock-ok}(c))
is a cover of $T$ of size $|F|$ with more twin links than $F$ 
($b's$ is a shadow of an existing link $bb'$). 
Furthermore, $F'$ can be modified to be shadows-minimal and/or smaller 
by replacing every link $tw$ with $t \in T_u$ and $w \notin T_u$ (if any) 
by the link $uw$ and removing redundant links.
This contradicts our choice of $F$.

Since $x \notin T_u$, and since $T_v$ is $a$-closed, $x$ is a proper ancestor of $u$.
We must have that $y=b'$, as otherwise one of $ax,by$ overlaps the other; see Fig.~\ref{f:lock-ok}(d).
Consider a link $x'z \in F$ that covers the edge between $x$ and its parent,
where $x' \in T_x$ and $z \notin T_x$.
Note that $x' \notin \{a,b,b'\}$, by Claim~\ref{c:shm}.
We must have $x'=x$, as otherwise, $x'z$ overlaps $ax$ or $bb'$.
Any rooted subtree that is $z$-closed contains $T_x$, and thus has at least $4$ leaves if $z$ is a leaf.
Hence $z$ cannot be a locked leaf.
\end{proof}

\begin{figure}
\centering 
\epsfbox{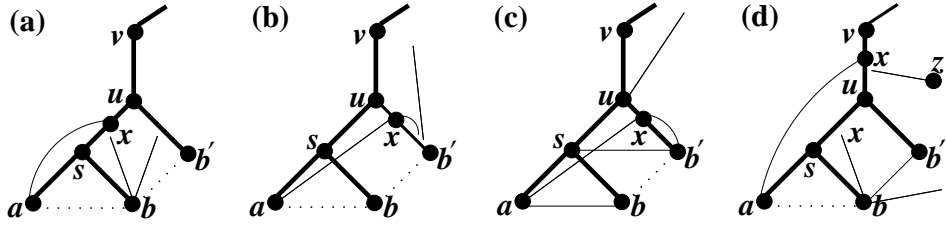}
\caption{Illustration to the proof of Claim~\ref{c:lock-ok}.
         Links in $F$ are shown by solid thin lines, twin and locking links (if not in $F$) 
         are shown by dotted lines. Some tree edges may be paths.}
\label{f:lock-ok}
\end{figure}

%%%%%%%%%%%%%%%%%%%%%%%%%%%%%%%%%%%%%
\section{The lower bound} \label{s:lb}
%%%%%%%%%%%%%%%%%%%%%%%%%%%%%%%%%%%%%

Let $S$ denote the set of stems of $T$ and let $X=V\setminus(L \cup S)$.

\begin{lemma} \label{l:lb}
Let $W$ be the set of twin and locking links,
$M$ a maximum matching in $E(L,L) \setminus W$, and $U$ the set of leaves unmatched by $M$.
Let $M_F=F(L,L) \setminus W$ and let
$N=\{bb' \in M: \mbox{each of } b,b' \mbox{ is unmatched by } M_F \}$.
Let $J$ be the set of links in $F$ not incident to a locked leaf. Then:
\begin{equation} \label{e:lb}
\frac{3}{2}|F| \geq \frac{3}{2}|M|+|U|+\frac{1}{2}|N|+\frac{1}{2}\sum_{x \in X} d_J(x) \ .
\end{equation}
\end{lemma}
\begin{proof}
Define a weight function $w$ on $E(L,V)$ by:
$$ 
w(e) = \left\{ 
\begin{array}{lll} 
3/2  & \mbox{ if } e \in E(L,L) \setminus W  \\
2    & \mbox{ if } e \in W \\ 
1    & \mbox{ if } e \in E(L,V \setminus L)  
\end{array} 
\right . 
$$
We prove the following two inequalities, that imply inequality (\ref{e:lb}):
\begin{eqnarray}
\frac{3}{2}|F| & \geq & w(F(L,V))+\frac{1}{2}\sum_{x \in X} d_J(x) \label{e:F1} \\
w(F(L,V))      & \geq & \frac{3}{2}|M|+|U|+\frac{1}{2}|N|          \label{e:F2}
\end{eqnarray}

We prove (\ref{e:F1}). Assign $3/2$ tokens to every $e \in F$, so there are $\frac{3}{2}|F|$ tokens.
We will show that these tokens can be reassigned such that:
every link in $F(L,L) \setminus W$ keeps its $3/2$ tokens, 
every link in $W \cap F$ gets $2$ tokens,
every link in $F(L,V \setminus L)$ keeps $1$ token from its $3/2$ initial tokens, and 
every $x \in X$ gets $1/2$ token for each link in $J$ incident to $x$.
Such an assignment is achieved as follows.
For every $e \in F$, move $1/2$ token from the $3/2$ tokens of $e$ to each {\em non-leaf} endnode of $e$, if any. 
Note that every link in $F(L,L)$ keeps its $3/2$ tokens, since no token is moved to leaves. 
In particular, every link $e \in W \cap F$ keeps its $3/2$ tokens.
We will assign to each $e \in W \cap F$
additional $1/2$ token moved earlier to some non-leaf node by some other link, as follows. 
\begin{itemize}
\item[$\bullet$]
Suppose that $e$ is a twin link with stem $s$.
By Claim~\ref{c:twins}, there is a unique link in $F$ incident to $s$, 
hence $s$ has $1/2$ token from this link.
This $1/2$ token is moved to $e$ to a total of $2$ tokens.
Note that after this $s$ has $0$ tokens.
\item[$\bullet$]
Suppose that $e$ is a locking link of a leaf $a$.
By Claim \ref{c:shm} there is a unique link $ax \in F$ incident to $a$,
and $x \in X$, by Claim~\ref{c:lock-ok}; 
hence $x$ has $1/2$ token from this link.   
This $1/2$ token is moved to $e$ to a total of $2$ tokens.
Note that after this $x$ has no tokens from the link $ax$
(but $x$ still has $1/2$ token from some other link $xz$, by Claim~\ref{c:lock-ok}).
\end{itemize}

We prove (\ref{e:F2}). 
Consider the graph $G=(L,E(L,L) \setminus W)$ and the graph $G'$ obtained from $G$ 
by removing the endnodes of the links in $N$. 
Since $M$ is a maximum matching in $G$ and since $N \subseteq M$, 
$M \setminus N$ is a maximum matching in $G'$. 
On the other hand, by the definition of $N$, no link in $M_F$ has a common endnode 
with a link in $N$, and thus $M_F$ is a matching in $G'$.
This implies $|M_F| \leq |M \setminus N|$, and since $N \subseteq M$ we get:
$$
|M_F| \leq |M|-|N| \ .
$$
By Claim~\ref{c:shm}, $F(L,L)$ is a matching, hence by the definition of $w$ we have: 
$$
w(F(L,V)) = \frac{3}{2}|M_F|+2|F(L,L) \cap W|+(|L|-2|M_F|-2|F(L,L) \cap W|)=|L|-\frac{1}{2}|M_F| \ .
$$
Combining with $|M_F| \leq |M|-|N|$ and observing that $|U|=|L|-2|M|$ we get
$$
w(F(L,V)) \geq |L|-\frac{1}{2}(|M|-|N|)=(|L|-2|M|)+\frac{3}{2}|M|+\frac{1}{2}|N|=|U|+\frac{3}{2}|M|+\frac{1}{2}|N|
$$
as claimed in (\ref{e:F2}).
\end{proof}

We prove the following statement that implies Theorem~\ref{t:main}.

\begin{theorem} \label{t:main+}
There exists a polynomial time algorithm that given an instance of {\sf TAP} computes a 
solution $I$ of size at most the right hand side of (\ref{e:lb}).
Thus $|I| \leq 1.5 \cdot |F|$.
\end{theorem}

%%%%%%%%%%%%%%%%%%%%%%%%%%%%%%%%%%%%%%%%%%%%%%%%%%%%%%%%%%%%%%%%%%%
\section{The credit scheme and the algorithm} \label{s:credit-algo}
%%%%%%%%%%%%%%%%%%%%%%%%%%%%%%%%%%%%%%%%%%%%%%%%%%%%%%%%%%%%%%%%%%%

We now explain how we use the lower bound (\ref{e:lb}).
We view the lower bound as a credit which we can spend for adding links to a partial solution $I$.
Initially, the algorithm assigns units of credit 
to nodes of $T$ and links in $M$ according to the four parts of the lower bound.
We call these credit units {\em coupons} and {\em tickets},
where each ticket worth half coupon. The initial credit distribution is as follows: 
\begin{itemize}
\item[$\bullet$]
$3/2$ coupons to every link $e \in M$ and $1$ coupon to every unmatched leaf $u \in U$.
\item[$\bullet$]
$1$ ticket to every link $e \in N$ and $d_J(x)$ tickets to every node $x \in X$.
\end{itemize}
The main difference between coupons and tickets is as follows. 
The location of coupons is known to us right after the matching $M$ is computed.
The location of tickets depends on $F$, which is not known to us, 
hence to ``claim'' a ticket, we will need to prove its existence. 
 
To {\em contract} a subtree $T'$ of $T$ is to combine all nodes in $T'$ into a single node $v$.
The edges and links with both endpoints in $T'$ are deleted.
The edges and links with one endpoint in $T'$ now have $v$ as their new endpoint.
Among any set of parallel links, if any, only one link is kept.
We refer to the nodes created by contraction as 
{\em compound nodes}; compound nodes always own $1$ coupon. 
For technical reasons, $r$ is also considered as a compound node.
Non-compound nodes of $T/I$ are referred to as {\em original nodes}. 
If we add a link $uv$ to a partial solution $I$, then the nodes along the path $P(uv)$ belong
to the same 2-edge-connected component of the augmented graph $(V,{\cal E} \cup I)$. 
Hence, we may contract some or all the edges of $P(uv)$. 
For a set of links $I \subseteq E$, let $T/I$ denote the tree obtained by
contracting every 2-edge-connected component of $T \cup I$ into a single node. 
We refer to the contraction of every $2$-edge-connected component of 
$T \cup I$ into a single node simply as the contraction of the links in $I$.
Let $T'$ be a subtree of $T/I$.
For a set $Y$ of links we use the notation 
$Y(T')$ or $Y \cap T'$ to denote the set of the links in $Y$ with both ends in $T'$.
If $Y$ is a set of nodes, then a similar notation is used to denote the set of the nodes in $Y$ 
that belong to $T'$.
We use the following notation for the credit distributed in $T'$:
\begin{itemize}
\item[$\bullet$]
$coupons(T')$ denotes the total number of coupons owned by $T'$: 
$1$ coupon for every unmatched leaf and every compound node of $T'$,
and $3/2$ coupons for every link in $M(T')$.
\item[$\bullet$]
$tickets(T')=|N(T')|+\sum_{x \in X(T')} \deg_J(x)$
denotes the number of tickets in $T'$.
\item[$\bullet$]
$credit(T')=coupons(T')+ \frac{1}{2} tickets(T')$. 
\end{itemize}

The algorithm maintains the following invariant.

\begin{invariant} [(Credit Invariant)] \label{inv:coupons} \ 
\begin{itemize}
\item[{\em (i)}]
Every unmatched leaf and every (unmatched) compound node of $T/I$ (including $r$) owns $1$ coupon, 
and every link in $M$ owns $3/2$ coupons.
\item[{\em (ii)}]
Every link in $N$ has a ticket, every original node $x \in X$ of $T/I$ has $d_J(x)$ tickets,
and compound nodes have no tickets.
\end{itemize}
\end{invariant}

The algorithm starts with a partial solution $I=\emptyset$ and with $credit(T/I)=credit(T)$
being the right-hand side of (\ref{e:lb}) plus $1$.
It iteratively finds a subtree $T'$ of $T/I$ and a cover $I'$ of $T'$, and 
{\em contracts $T'$ with $I'$}, which means the following:
add $I'$ to $I$, contract $T'$, and assign $1$ coupon (and $0$ tickets) to the new compound node.  
To use the notation $T/I$ properly, we will assume that $I'$ is an {\em exact cover} of $T'$, 
namely, that the set of edges of $T/I$ that is covered by $I'$ equals the set of edges of $T'$
(this is possible due to shadow completion).
Let us say that a contraction of $T'$ with $I'$ is {\em legal} if 
$credit(T') \geq |I'|+1$.
This means that the set $I'$ of links added to $I$ and 
the $1$ coupon assigned to the new compound node are paid by the total credit in $T'$.
The credit of $T'$ is not reused in any other way, since the only credit the new 
compound node has is the $1$ coupon assigned to it, and it has no tickets.
We do only legal contractions, which implies that at any step of the algorithm 
$$|I|+credit(T/I) \leq credit(T) \ .$$
Thus at the last iteration, when $T/I$ becomes a single compound node, 
$|I|$ is at most the right-hand side of (\ref{e:lb}).

We now describe two legal contractions of $T'$ with $I'$ that rely on coupons only.

\begin{definition} [(greedy contractions)] 
The following two types of contractions of $T'$ with $I'$ are called {\em greedy contractions}:
\begin{itemize}
\item[$\bullet$] 
{\em Greedy locking tree contraction}: 
Here $T'$ is a locking tree of $a$ (as in Definition~\ref{d:lock}),
$a,b,b'$ are all unmatched by $M$, and $I'=\{bb',up(a)\}$;
furthermore, if $b$ is also a locked leaf then $T'$ contains the locking tree of $b$ 
(see Fig.~\ref{f:twin-lock}(c) where the locking tree $T_v$ of $a$ contains the locking tree $T_v'$ of $b$).  
Note that $I'$ indeed covers $T'$ and that $coupons(T') \geq 3$, hence this contraction is legal. 
\item[$\bullet$] 
{\em Greedy link contraction}: 
Here $T'=P(uv)$ for some $uv \in E$ where $u,v$ are unmatched leaves of $T/I$,
and $I'=\{uv\}$.
Note that $coupons(T') \geq 2$, hence this contraction is legal.
\end{itemize}
\end{definition}

The first step of our algorithm is exhausting all greedy locking tree contractions.
We now describe a more complicated type of legal contraction used in \cite{EFKN-APPROX}
(but our definitions are slightly different).
For a node set $U \subseteq V$, we let $up(U)=\{up(u):u \in U\}$.
A rooted subtree $T'$ of $T$ is {\em $U$-closed} 
if there is no link in $E$ from $U \cap T'$ to $T \setminus T'$.
$T'$ is {\em leaf-closed} if it is $L(T)$-closed.
A leaf-closed $T'$ is {\em minimally leaf-closed}
if any proper rooted subtree of $T'$ is not leaf-closed.

\begin{definition} 
[($M$-compatible tree, semi-closed tree)] \label{d:semi-closed}
Let $M$ be a ma\-tching on the leaves of $T/I$.
A subtree $T'$ of $T/I$ is {\em $M$-compatible} if
for any $bb' \in M$ either both $b,b'$ belong to $T'$, or none of $b,b'$ belongs to $T'$.
A rooted subtree $T'$ of $T/I$ is {\em semi-closed} 
(w.r.t. $M$) if it is $M$-compatible and closed w.r.t. its unmatched leaves.
$T'$ is {\em minimally semi-closed} if $T'$ is semi-closed but any proper subtree of $T'$ is not semi-closed.
\end{definition}

For a semi-closed tree $T'$ let us use the following notation:
\begin{itemize}
\item[$\bullet$]
$M'=M(T')$ is the set of links in $M$ with both endnodes in $T'$.
\item[$\bullet$]
$U'=U(T')$ is the set of unmatched leaves of $T'$.
\end{itemize}

\begin{lemma} \label{l:up}
If $T'$ is minimally semi-closed then $M' \cup up(U')$ is an exact cover of $T'$.
\end{lemma}
\begin{proof}
Let $T''$ be obtained from $T'$ by contracting $M'$. 
Note that $L(T'')=U'$. 
Otherwise, if $T''$ has a leaf $a$ that is not a leaf of $T'$, then the subtree of $T'$ 
that was contracted into $a$ is a semi-closed tree (with no unmatched leaves),
contradicting the minimality of $T'$. 
Note also that $T''$ is minimally leaf-closed, since $T'$ is minimally semi-closed. 
\cite{N} proved that if $T''$ is a minimally leaf-closed tree
then $up(L(T''))$ is an exact cover of $T''$.
Thus $up(U')$ is an exact cover of $T''$ (if $T'$ has no unmatched leaves then $T''$ is a single node). 
As a link in $M'$ can cover only edges in $T'$, the statement follows.
\end{proof}

Thus a minimally semi-closed tree admits a cover of size $|M'|+|U'|$. 
This motivates the following definition. 

\begin{definition} [(deficient tree)] 
A semi-closed tree $T'$ is {\em deficient} if \\ $credit(T')< |U'|+|M'|+1$.
\end{definition}

Recall that for a cover $I'$ of a subtree $T'$ of $T/I$, 
{\em contracting $T'$ with $I'$} means that we 
add $I'$ to $I$, contract $T'$, and assign $1$ coupon to the new compound node.  
The main actions taken by our algorithm can be summarized as follows.

\begin{invariant} [(Partial Solution Invariant)] \label{inv:I} 
The partial solution $I$ is obtained by initially exhausting greedy locking tree contractions,
and then sequentially applying a greedy link contraction, or 
contracting a semi-closed tree with an exact cover.
\end{invariant}

In the next sections we prove the following key lemma.

\begin{lemma} \label{l:B'}
Suppose that the Credit Invariant and the Partial Solution Invariant hold for $T$, $M$, and $I$
(in particular, greedy locking tree contractions were initially exhausted), 
that $M$ has no twin links and no locking links, and that $T/I$ has no greedy link contraction.
Then there exists a polynomial time algorithm that finds a non-deficient semi-closed tree $T'$ of $T/I$ 
and an exact cover $I' \subseteq E$ of $T'$ of size $|I'|=|M'|+|U'|$. 
\end{lemma}

Algorithm {\sc Tree-Cover} (Algorithm~\ref{alg:F}) initiates $I \gets \emptyset$ as a partial cover.
It computes a maximum matching $M$ in $E(L,L) \setminus W$
and distributes coupons as described in Credit Invariant~\ref{inv:coupons}(i).
Then it exhausts greedy locking tree contractions.
In the main loop, the algorithm iteratively exhausts greedy link contractions,
then finds $T',I'$ as in Lemma~\ref{l:B'}, and contracts $T'$ with $I'$.
The stopping condition is when $I$ covers $T$, namely, when $T/I$ is a single node.

\medskip

\begin{algorithm}[H]
\caption{{\sc Tree-Cover}$(T=(V,{\cal E}),E)$ (A $1.5$-approximation algorithm)} \label{alg:F}
{\bf initialize:}  $I \gets \emptyset$  \\
$M \gets$ maximum matching in $E(L,L) \setminus W$. \\
Assign $1$ coupon to every unmatched leaf and to $r$, and $3/2$ coupons to every link in $M$. \\
Exhaust greedy locking tree contractions. \\
\While{$T/I$ has more than one node}
{
Exhaust greedy link contractions and update $I$ and $M$ accordingly. \\
Find a subtree $T'$ of $T/I$ and an exact cover $I'$ of $T'$ as in Lemma~\ref{l:B'}. \\
Contract $T'$ with $I'$.
}
\Return{$I$}
\end{algorithm}

\medskip

It is easy to see that all the steps in the algorithm can be 
implemented in polynomial time and that during the algorithm
the Credit Invariant and the Partial Solution Invariant hold for $T$, $M$, and $I$.
The credit scheme used implies that the algorithm computes a 
solution $I$ of size at most $1.5$ times the right-hand size of~(\ref{e:lb}).
Hence it only remains to prove Lemma~\ref{l:B'}, which is done in the rest of the paper.
 
We note that our paper uses the main idea of a correct and relatively simple proof 
of the $1.8$ ratio in \cite{EFKN-TALG}.
In fact, a slight modification of this algorithm gives ratio $1.75$, see \cite{KN-LP}. 
However, the proof of the $1.5$ ratio is much more involved, and we 
mention the relation of our current paper to the previous incomplete/incorrect proofs of the $1.5$ ratio
\cite{EFKN-APPROX} and \cite{EKN-IPL}, which are co-authored by the authors of the current paper.
One major difference is the definition of a locked leaf.
Without going into details, the \cite{EFKN-APPROX} definition of locked leaf
leads to several additional complex structures and definitions, 
and to an exhaustive case analysis of many deficient minimally semi-closed trees with $3$, $4$, and $5$ leaves. 
In this paper we essentially have just one deficient tree -- with $3$ leaves;
two additional $4$-leaf trees are reduced to the $3$-leaf case. 
While we do not see an explicit mistake in the \cite{EFKN-APPROX} proof line,
our attempt to write a full version resulted in a very complex paper with more than 40 pages. 
On the other hand, the locked leaf definition in \cite{EKN-IPL} is erroneous, as it leads to an 
improper claiming of tickets, as was brought to our attention recently by \cite{CGLS}. 
This also leads to an additional error of not using the term $\frac{1}{2}|N|$ in the lower bound,
while it is essential for the proof of the $1.5$ ratio. 
Additional simplifications in our current paper are:  a simpler proof of the lower bound, 
a much easier case analysis than in the \cite{EFKN-APPROX} full draft,  
removal of various greedy steps and preprocessing reductions, and more.

%%%%%%%%%%%%%%%%%%%%%%%%%%%%%%%%%%%%%%%%%%%%%%%%%%%%%
\section{Proof of Lemma~\ref{l:B'}} \label{s:B'}
%%%%%%%%%%%%%%%%%%%%%%%%%%%%%%%%%%%%%%%%%%%%%%%%%%%%%

%%%%%%%%%%%
\subsection{Dangerous trees} \label{ss:semi-prop}
%%%%%%%%%%%

To prove Lemma~\ref{l:B'} we will give a characterization of deficient trees, 
by establishing that $T'$ is deficient iff the graph formed by $T'$ 
and the links in $F$ that have an endnode in $T'$ has a certain ``bad'' structure.
But even having such a characterization does not achieve the goal of Lemma~\ref{l:B'}. 
One reason is that this characterization depends on $F$, so we
are not able to recognize in polynomial time whether a given $T'$ is indeed deficient.
We thus classify a tree $T'$ as ``dangerous'' if $T'$ and the links in $E$ incident to nodes of $T'$ 
contain such a ``bad'' structure; thus a non-dangerous tree cannot be deficient. 

Let $T'$ be a semi-closed tree with root $v$. 
In addition to the notation $M',U'$ established in the previous section, let us use the following notation:
\begin{itemize}
\item[$\bullet$]
$C'$ is the set of {\em non-leaf} compound nodes of $T'$ (this includes $r$, if $r \in T'$).
\item[$\bullet$]
$L'=L(T')$ is the set of leaves of $T'$. 
\item[$\bullet$]
$S'=S(T')$ is the set of stems of $T'$. 
\end{itemize}

We consider the following family of semi-closed trees,
that can be recognized in polynomial time, and (as we will show) includes the deficient trees.

\begin{definition} [(dangerous tree)] \label{d:dangerous}
A semi-closed tree $T'$ is called {\em dangerous} if 
$|C'|=0$, $|M'|=1$, and one of the following holds:
\begin{itemize}
\item[{\em (i)}]
$|L'|=3$, $|S'|=0$, and $T'$ is as in Fig.~\ref{f:def}(a) with the links depicted present~in~$E$.
Namely, if $a$ is the unmatched leaf of $T'$, then there exists an ordering $b,b'$ of the matched leaves 
of $T'$ such that $ab' \in E$, the contraction of $ab'$ does not create a new leaf, and $T'$ is $b$-open.
If such an ordering $b,b'$ is not unique 
(namely, if also $ab \in E$, the contraction of $ab$ does not create a new leaf, 
and $T'$ is $b'$-open -- see Fig.~\ref{f:def}(b)), then 
we will assume that the up-node of $b$ is an ancestor of the up-node of $b'$.
\item[{\em (ii)}]
$|L'|=4$, $|S'|=1$, say $S'=\{s\}$, 
exactly one of the twins of $s$ is matched by $M$, and the tree $\tilde{T}'$ obtained from $T'$ 
by contracting the twin link of $s$ is a $3$-leaf dangerous tree, see Fig.~\ref{f:def}(c,d,e).
\end{itemize}
\end{definition}

\begin{figure}
\centering 
\epsfbox{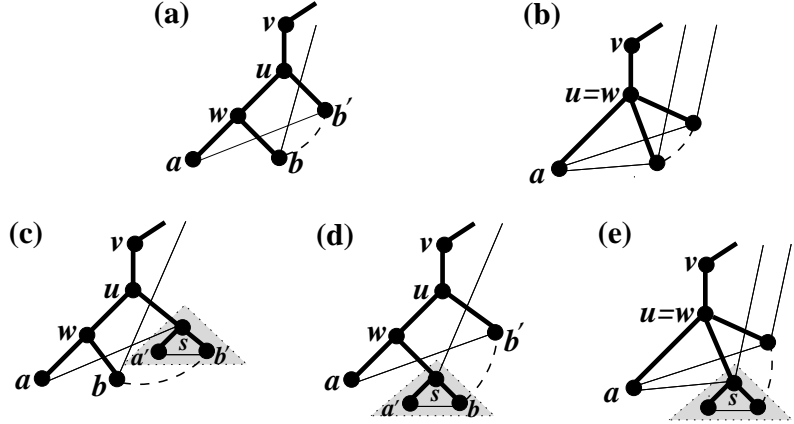}
\caption{Dangerous trees.
         The dashed arc shows the matched pair $bb'$. 
         Solid thin lines show links in $E$.
         Here $s$ is a stem and $w$ is not a stem. 
         Some of the edges of $T'$ can be paths, and $u=v$ or/and $u=w$ may hold.
         }
   \label{f:def}
\end{figure}

In Section~\ref{s:char} we will prove the following key statement. 

\begin{lemma} \label{l:char}
Under the assumptions of Lemma~\ref{l:B'}, any deficient tree is dange\-rous.
\end{lemma}

Note that a dangerous tree $T'$ may {\em not} be deficient. 
E.g., if $T'$ is a $3$-leaf dangerous tree as in Fig.~\ref{f:def}(a),
then we may have that $F'=\{bb',av\}$; 
in this case, $T'$ is dangerous (if the links in in Fig.~\ref{f:def}(a) present in $E$) 
but $T'$ is not deficient (since $T'$ has a ticket at $v$, and thus $credit(T') \geq 3$, while $|M'|+|U'|=2$). 
The property of $3$-leaf dangerous trees that we will use 
is that the links $ab'$ and $bz$ with $z \notin T'$ {\em exist in $E$},
but we do not care whether they belong to $F$ or not.

Note that the property of a tree being dangerous depends only on the structure of
the tree and existence/absence of certain links in $E$ and $M$, 
and thus can be tested in polynomial time.
If we find a minimally semi-closed tree that is not dangerous then we are fine.
However, it might happen that all minimally semi-closed trees are dangerous (and they even may be all deficient).
Using the structure of dangerous trees that guarantees existence of certain links that go ``outside'' $T'$, 
we will show in the next section how to find a non-minimal semi-closed tree $T'$
that still admits a relatively small cover of size $|M'|+|U'|$, but is not dangerous.

%%%%%%%%%%%
\subsection{Finding a good tree when all minimally semi-closed trees are dangerous} \label{ss:all-dangerous}
%%%%%%%%%%%

We use Lemma~\ref{l:char} to prove Lemma~\ref{l:B'}, 
by showing that if all minimally semi-closed trees are dangerous,
then we can find a non-minimal non-deficient semi-closed tree $T'$ and its exact cover 
$I'$ such that $|I'|=|M'|+|U'|$.
Let ${\cal D}$ denote the family of minimally semi-closed subtrees of $T/I$.
Clearly, the trees in ${\cal D}$ are pairwise node disjoint, and can be found in polynomial time,
since every $D \in {\cal D}$ is a rooted subtrees of $T/I$, and since we can check 
polynomial time whether a subtree of $T/I$ is semi-closed. 
We also can check in polynomial time whether a member of ${\cal D}$ is dangerous.
If there is $T' \in {\cal D}$ that is not dangerous, then $T'$ is not deficient,
so $T'$ and $I'=M' \cup up(U')$ satisfy the requirement of Lemma~\ref{l:B'}.
Thus we will assume that all the trees in ${\cal D}$ are dangerous.
We will show that then Algorithm~\ref{alg:tF} finds a
non-deficient semi-closed tree $T'$ and its exact cover $I'$ such that $|I'|=|M'|+|U'|$.

In Algorithm~\ref{alg:tF} we define a new tree $\tilde{T}$ obtained from $T/I$ by contracting the twin link
in every $4$-leaf tree $D \in {\cal D}$; 
this transforms every such $D$ into a $3$-leaf dangerous tree $\tilde{D}$.
Then we {\em temporarily consider} a matching $\tilde{M}$ on the leaves of $\tilde{T}$ obtained from $M$ 
by replacing the link $bb'$ by the link $ab'$ in each $3$-leaf dangerous tree.
We emphasize that $\tilde{M}$ is considered only for the purpose of finding the pair $T',I'$, 
but we {\em do not replace $M$ by $\tilde{M}$}. 
Note that the property of a tree being semi-closed or dangerous depends on the matching.
In what follows, ``dangerous'' always means w.r.t. the matching $M$;
for ``semi-closed'' the default matching is $M$, and we will specify each time when 
a tree is semi-closed w.r.t. the matching $\tilde{M}$.

\medskip

\begin{algorithm}[H]
\caption{{\sc Find-Tree}$(T=(V,{\cal E}),E,M)$ (Finds a non-dangerous semi-closed tree $T'$ and 
its exact cover $I'$ of size $|I'|=|U'|+|M'|$, when all minimally semi-closed trees are dangerous.)}
\label{alg:tF}
Let $\tilde{W}$ be the link set obtained by picking for every $4$-leaf tree 
$D \in {\cal D}$ the twin-link of the stem of $D$ (see Fig.~\ref{f:d}(a)). \\
Let $\tilde{T}=(T/I)/\tilde{W}$ be obtained from $T/I$ by contracting every link in $\tilde{W}$, and 
let $\tilde{\cal D}$ be the set of subtrees of $\tilde{T}$ that correspond to ${\cal D}$ (see Fig.~\ref{f:d}(b)). 
\newline
$\triangleright$ 
{\small Comment: No link in $M$ is contracted, and $M$ is a matching on the leaves of $\tilde{T}$.
Every $\tilde{D} \in \tilde{\cal D}$ is a $3$-leaf dangerous tree, by the definition of a $4$-leaf dangerous tree.} 
\\
Let $\tilde{M}$ be obtained from $M$ by ``switching'' 
the link $e=bb'$ by the link $\tilde{e}=ab'$ in every $\tilde{D} \in \tilde{\cal D}$ (see Fig.~\ref{f:d}(c)),
where $a,b,b'$ are as in Definition~\ref{d:dangerous}(i). 
\newline
$\triangleright$ 
{\small Comment: $\tilde{M}$ is also a matching on the leaves of $\tilde{T}$.}
\\
Let $\tilde{T}'$ be a minimally semi-closed tree in $\tilde{T}$ w.r.t. $\tilde{M}$ 
(see the trees $\tilde{T}'_1$ and $\tilde{T}'_2$ in Fig.~\ref{f:d}(c)). \\
Let $\tilde{I}'=\tilde{M}(\tilde{T}') \cup up(\tilde{U}')$, where $\tilde{U}'$ is the set of 
$\tilde{M}$-unmatched leaves of $\tilde{T}'$. \\
Let $\tilde{W}'$ be the set of twin links in $\tilde{W}$ contained in the leaves of $\tilde{T}'$, and let 
$T'$ be obtained from $\tilde{T}'$ by ``uncontracting'' the links in $\tilde{W}'$ (see Fig.~\ref{f:d}(d)).
\\
\Return{$T'$ and $I'=\tilde{I}' \cup \tilde{W}'$}
\end{algorithm}

\begin{figure}
\centering 
\epsfbox{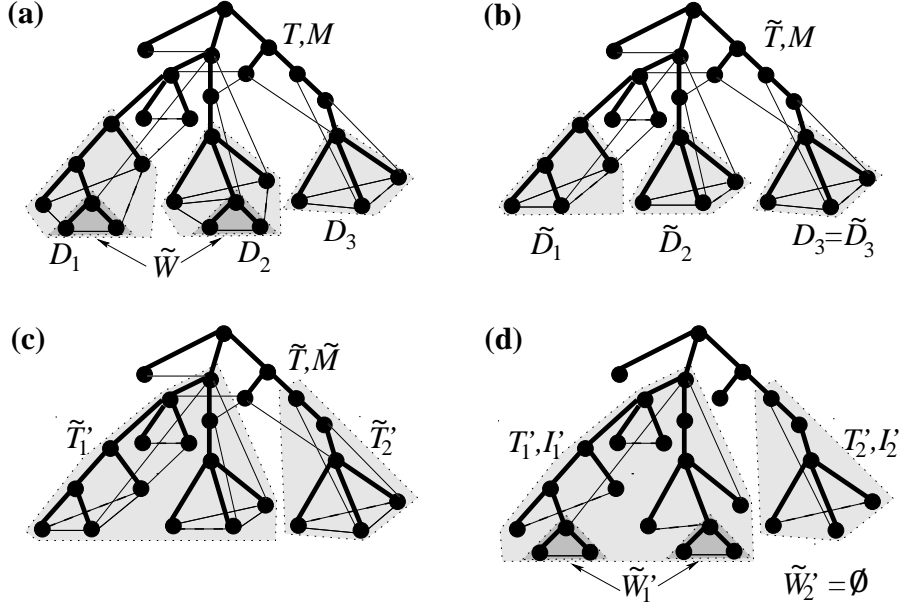}
\caption{Illustration to Algorithm~\ref{alg:tF}: 
         Finding a good tree when all minimally semi-closed trees are dangerous.
         The dashed lines show links in the matching and solid thin lines show links in $E$.}
   \label{f:d}
\end{figure}

\medskip

We prove that the pair $T',I'$ computed by the algorithm satisfies the following:
\begin{itemize}
\item[$\bullet$]
$T'$ is semi-closed and not dangerous (thus $credit(T') \geq |M'|+|U'|+1$).
\item[$\bullet$]
$I'$ is an exact cover of $T'$ of size $|I'|=|M'|+|U'|$.
\end{itemize}

Consider the tree $\tilde{T}'$ and its cover $\tilde{I}'$ computed at lines 4 and 5 of the algorithm.

\begin{claim} \label{c:larger-tree}
Suppose that $\tilde{T}'$ contains a $3$-leaf dangerous tree $\tilde{D}'$ 
with leaves $a,b,b'$ as in Definition~\ref{d:dangerous}. 
If $\tilde{T}'$ is $b$-closed, then $\tilde{T}'$ is not dangerous.
\end{claim}
\begin{proof}
Note that no $4$-leaf dangerous tree has a $3$-leaf dangerous tree as a rooted subtree. 
This is since a $4$-leaf dangerous tree has a stem, while a $3$-leaf dangerous tree has no stems.
Thus if $\tilde{T}'$ has at least $4$ leaves then $\tilde{T}'$ is not dangerous.
If $\tilde{T}'$ has $3$ leaves, then the leaf set of $\tilde{T}'$ and $\tilde{D}'$ coincide.
Thus the only possibility for $\tilde{T}'$ to be dangerous is if we are in the case of Fig.~\ref{f:def}(b).
However, in Definition~\ref{d:dangerous}(i) we assume
that the up node of $b$ is an ancestor of the up-node of $b'$, 
which implies that $T'$ is both $b$-closed and $b'$-closed. Thus $\tilde{T}'$ is not dangerous.
\end{proof}

\begin{claim} \label{c:disjoint}
For any $\tilde{D} \in \tilde{\cal D}$,
either $\tilde{T}'$ properly contains $\tilde{D}$ or $\tilde{T}',\tilde{D}$ are node disjoint.
Furthermore, if $\tilde{T}'$ contains some $\tilde{D}' \in \tilde{\cal D}$ then $\tilde{T}'$ is not dangerous.
\end{claim}
\begin{proof}
As $\tilde{T}',\tilde{D}$ are rooted subtree of $\tilde{T}$, 
they are either node disjoint or one of them contains the other. 
Assume that $\tilde{T}',\tilde{D}$ are not node disjoint.
Then they share a leaf.
Suppose that the dangerous tree $\tilde{D}$ and its leaves $a,b,b'$ 
are as in Definition~\ref{d:dangerous}(i).
Note that we must have $b \in \tilde{T}'$,
since $a \in \tilde{T}'$ or $b' \in \tilde{T}'$ implies $ab' \in \tilde{T}'$
(since $ab' \in \tilde{M}$ and since $\tilde{T}'$ is semi-closed w.r.t. $\tilde{M}$),
hence the least common ancestor of $a,b'$, and all its descendants also belong to $\tilde{T}'$.
Since $\tilde{T}'$ is $b$-closed and $\tilde{D}$ is $b$-open, $\tilde{T}'$ properly contains $\tilde{D}$.
The second statement follows from Claim~\ref{c:larger-tree} and the fact that 
every $\tilde{D} \in \tilde{\cal D}$ is a $3$-leaf dangerous tree.
\end{proof}

Let $h:M \longrightarrow \tilde{M}$ be defined by $h(e)=\tilde{e}$ if $e \in M \setminus \tilde{M}$ 
and $h(e)=e$ otherwise, where $\tilde{e}$ is as in Line~3 in Algorithm~\ref{alg:tF}.
Since the trees in $\tilde{\cal D}$ are pairwise node disjoint, $h$ is a bijection. 
Now we prove the following.

\begin{claim} \label{c:f}
For any $e \in M$, either each of $e,h(e)$ has both endnodes in $\tilde{T}'$, 
or none of $e,h(e)$ has an endnode in $\tilde{T}'$;
thus $\tilde{T}'$ is $M$-compatible.
\end{claim}
\begin{proof}
Let $e \in M$.
If $h(e)=e$ then the statement holds, so assume that $h(e)=\tilde{e}=ab'$ is as in Line~3 
in Algorithm~\ref{alg:tF}, 
and $\tilde{D} \in \tilde{\cal D}$ is the corresponding $3$-leaf tree with leaves $a,b,b'$.
By Claim~\ref{c:disjoint}, either $\tilde{T}'$ properly contains $\tilde{D}$, 
or $\tilde{T}',\tilde{D}$ are node disjoint.
In the former case, each of $e,h(e)$ has both endnodes in $\tilde{T}'$, 
while in the later case none of $e,h(e)$ has an endnode in $\tilde{T}'$.
In particular, any $e \in M$ either has both endnodes in $\tilde{T}'$ or has no 
endnode in $\tilde{T}'$; hence $\tilde{T}'$ is $M$-compatible.
\end{proof}

\begin{claim} \label{c:size}
$\tilde{T}'$ is semi-closed (w.r.t. to $M$), 
$\tilde{I}'$ is an exact cover of $\tilde{T}'$,
and $|\tilde{M}(\tilde{T}')|=|M(\tilde{T}')|$.
\end{claim}
\begin{proof}
By Claim~\ref{c:f}, $\tilde{T}'$ is $M$-compatible. Hence to show that 
$\tilde{T'}$ is semi-closed we need to show that $\tilde{T}'$ 
is $a$-closed for any its leaf $a$ unmatched by $M$.
If $a$ is unmatched by $\tilde{M}$, then this is so since $\tilde{T'}$ is semi-closed w.r.t. $\tilde{M}$.
Otherwise, $a$ is a leaf in a $3$-leaf dangerous tree $\tilde{D}$ as in Definition~\ref{d:dangerous}(i),
and $\tilde{T}'$ contains $\tilde{D}$, by Claim~\ref{c:disjoint}.
As $\tilde{D}$ is $a$-closed, so is $\tilde{T}'$.
Applying Lemma~\ref{l:up} on $\tilde{T}'$ and $\tilde{M}$, 
we get that $\tilde{I}'$ is an exact cover of $\tilde{T}'$.
The equality $|\tilde{M}(\tilde{T}')|=|M(\tilde{T}')|$
follows from Claim~\ref{c:f}.
\end{proof}

Now let us consider the pair $T',I'$ returned by the algorithm.  
Recall that $\tilde{T}'$ is obtained from $T'$ by contracting the links in $\tilde{W}'$,
where $\tilde{W}'$ is the set of links in $\tilde{W}$ with both endnodes in $T'$.

\begin{claim}
$T'$ is semi-closed and not dangerous.
\end{claim}
\begin{proof}
Since $\tilde{T}'$ is semi-closed, so is $T'$.
This implies that $T'$ contains some $D' \in {\cal D}$, hence $\tilde{T}'$ contains $\tilde{D}'$.
By Claim~\ref{c:disjoint}, $\tilde{T}'$ is not dangerous. 
This implies that $T'$ is not dangerous.
\end{proof}

\begin{claim}
$I'$ is an exact cover of $T'$ of size $|I'|=|M'|+|U'|$.
\end{claim}
\begin{proof}
By Lemma~\ref{l:up}, $\tilde{I}'$ is an exact cover of $\tilde{T}'$, 
and $\tilde{T}'$ is obtained from $T'$ by contracting $\tilde{W}'$. 
Thus $I'=\tilde{I}' \cup \tilde{W}'$ is an exact cover of $T'$.
To see that $|I'|=|M'|+|U'|$, note that by Claim~\ref{c:size} $|M'|=|\tilde{M}(T')|$,
and that $|\tilde{U}'|=|U'|-|\tilde{W}'|$, since 
the contraction of every link in $\tilde{W}'$ does not change $|M'|$ but reduces $|U'|$ by $1$.
Thus 
$$|I'|=|\tilde{I}'|+|\tilde{W}'|=|\tilde{U}'|+|\tilde{M}'|+|\tilde{W}'|=|U'|+|M'|$$
as claimed.
\end{proof}

The proof of Lemma~\ref{l:B'} is now complete. It remains only to prove Lemma~\ref{l:char}.

%%%%%%%%%%%%%%%%%%%%%%%%%%%%%%%%%%%%%%%%%%%%%%%%%%%%
\section{Proof of Lemma~\ref{l:char}} \label{s:char}
%%%%%%%%%%%%%%%%%%%%%%%%%%%%%%%%%%%%%%%%%%%%%%%%%%%%

\subsection{Overview}

The proof of Lemma~\ref{l:char} is long and non-trivial, so we will give its intuitive overview.
Let $T'$ be a semi-closed tree with root $v$. 
In addition to the notation $M',U',C',L',S'$ established in previous sections, let  
$X'=X(T')$ denote the set of (original) nodes in $X$ that belong to $T'$, and let
$$S'_1=\{s \in S':\mbox{the twin-link of } s \mbox{ is in } F\} \ .$$
Note that by Claim~\ref{c:twins} $S'_1=\{s \in S':d_F(s)=1\}$
and that $d_F(s)=0$ for every $s \in S' \setminus S'_1$.

In what follows assume that $T'$ is deficient, so $credit(T')-(|M'|+|U'|)<1$.
Note that $coupons(T')-(|U'|+|M'|) \geq \frac{1}{2}|M'|+|C'|$.
Thus we must have $|C'|=0$ and $|M'| \leq 1$.

Let us focus on the main case $|M'|=1$. Then $T'$ cannot have a ticket,
since the unique link in $M'$ already carries a surplus of $1/2$ credit unit, 
and any ticket gives another $1/2$ extra credit.
Furthermore, we will show that for any $x \in X'$, $d_F(x) \geq 1$ implies $d_J(x) \geq 1$, 
hence in the case $|M'|=1$ we do not need to worry about links incident to locked leaves 
(since just one ticket makes $T'$ non-deficient).

We will show that deficient trees with $|M'|=1$ are ``small'' -- have at most $4$ leaves.
To establish this, we look at the links that cover the set $U'$ of the unmatched leaves of $T'$.
No such link connects two leaves in $U'$, as we assume that all greedy link contractions are exhausted.
Thus there are at least $|U'|$ such links.
Each of these links has both endnodes in $T'$, since $T'$ is $U'$-closed.
In addition, there is a link that covers the edge between the root $v$ of $T'$ and the parent of $v$.
This link has no endnode in $U'$, since $T'$ is $U'$-closed.
We thus have a set of $|U'|+1$ links that have an endnode in $T' \setminus U'$.
To avoid a ticket, such links cannot have an endnode in $X'$.
It follows therefore that each of these $|U'|+1$ links is incident 
to a node in $S' \cup \{b,b'\}$, where $bb'$ is the unique link in $M'$. 
However, $d_F(y) \leq 1$ for every node $y \in S' \cup \{b,b'\}$ (by Claims \ref{c:shm} and \ref{c:twins}),
so we get that $|U'|+1 \leq 2|M'|+|S'|=2+|S'|$, and thus $|U'| \leq |S'|+1$.

Note that every stem has a leaf matched by $M$, hence $|S'| \leq 2|M'| \leq 2$.
We will show that if $|S'|=2$ then $|N(T')| \geq 1$, which gives us a ticket.
Consequently, we get that $|S'| \leq 1$.
This implies $|U'| \leq |S'|+1 \leq 2$, and since $|L'|=2|M'|+|U'| \leq |S'|+3$, 
we get that $|L'| \leq 4$, and if $|L'|=4$ then $|S'|=1$. 

We will show that contracting the twin link of the stem of a $4$-leaf deficient tree
results in a $3$-leaf deficient tree, and we pin down how $3$-leaf deficient trees with $|M'|=1$ look like.
Finally, we will exclude the case $|L'|=2$ by showing that in this case $T'$ must have a compound node.

We will refute the case $|M'|=0$ by showing that then $T'$ has $2$ tickets 
(this is the only place where we need to be careful about links incident to locked leaves).

The rest of this section is organized as follows. 
In the next Section~\ref{ss:tickets} we will explain how we claim tickets.
Then in Section~\ref{ss:M-inv} we will derive some properties of $T'$, 
and show that either $|M'|=|S'|=0$, or  
$|M'|=1$, $|S'| \leq 2$, and $|L'| \leq |S'_1|+3$, see Lemma~\ref{l:3-lvs}.
In Section~\ref{ss:char} we finish the proof of Lemma~\ref{l:char}.
We first show that if $|M'|=1$ then $|S'_1| \leq 1$ and thus $|L'| \leq 4$, see Claim~\ref{c:s'}.
Then we prove that if $|L'| \in \{3,4\}$ then $T'$ is dangerous. 
Finally, we refute the cases $|L'|=2$ and $|M'|=0$.

%%%%%%%%%%%
\subsection{Claiming tickets} \label{ss:tickets}
%%%%%%%%%%%

In order to prove Lemma~\ref{l:char} we will exclude some trees $T'$ by claiming that $T'$ has tickets.
Note that we do not need to specify the node or the link on which the ticket is claimed.
All we care about is that $T'$ has a ticket, for {\em any} possible choice of $F$.
To claim a ticket in a rooted subtree $T'$ of $T/I$, 
we prove that for {\em any} choice of $F$, one of the following must hold:
\begin{itemize}
\item[$\bullet$]
{\em $N$-ticket}: 
There is a link $bb' \in M$ with $b,b' \in T'$ 
such that each of $b,b'$ is an endnode of a twin link in $F$ (see Fig.~\ref{f:NX}(a)); 
note that indeed $bb' \in N$ ($N$ is defined in Lemma~\ref{l:lb}), 
since $d_F(a)=1$ for every original leaf $a$ (Claim~\ref{c:shm}).
\item[$\bullet$]
{\em $X$-ticket}: 
There is an {\em original node} $x \in X \cap T'$ and a link $e=xa' \in F$,
where $a'$ is a (possibly compound) node of $T/I$, such that 
the original endnode $a$ of $e$ contained in $a'$ (possibly $a=a'$) 
is not a locked leaf (see Fig.~\ref{f:NX}(b)).
\end{itemize}

\begin{figure} 
\centering 
\epsfbox{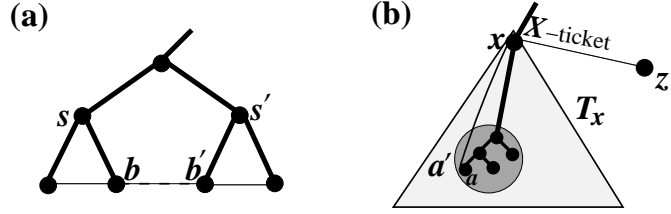}
\caption{Ticket types. Links in $F'$ are shown by solid thin lines, and links in $M$ by dashed lines.
        (a) The link $bb'$ has an $N$-ticket. 
        (b) The node $x \in X$ has an $X$-ticket (due to the link $xz$), 
            even if the original endnode $a$ of the link $a'x$ (here $a'$ is a compound node) is a locked leaf.
         }
   \label{f:NX}
\end{figure}

The following statement is useful for claiming an $X$-ticket.

\begin{claim} \label{c:X-link}
Let $I$ be an arbitrary link set and let $x \in X \cap T/I$. 
Let $e=a'x \in F$ where $a'$ is a node of $T/I$.
If the original endnode $a$ of $e$ contained in $a'$ (possibly $a=a'$) is a locked leaf
then $a' \in T_x$ and $x$ has a ticket for some link $xz$ with $z \notin T_x$.
\end{claim}
\begin{proof}
By Claim~\ref{c:lock-ok}, $x$ is an ancestor of $a$ in $T$, 
and there is a link $xz \in F$ such that $z$ is not a locked leaf. 
This implies that $x$ is an ancestor of $a'$ in $T/I$.
Furthermore, $x$ has a ticket for the link $xz$, since $z$ is not a locked leaf.
\end{proof}

Claim~\ref{c:X-link} has the following immediate two important consequences. 

\begin{corollary} \label{c:a-tick}
Let $I$ be an arbitrary link set and let $x \in X \cap T/I$. Then:
\begin{itemize}
\item[{\em (i)}]
If $d_F(x) \geq 1$ then $d_J(x) \geq 1$, and thus $x$ has an $X$-ticket.
\item[{\em (ii)}]
Every link $xz \in F$ with $z \notin T_x$ contributes an $X$-ticket to $x$.
\end{itemize}
\end{corollary}

Part~(i) of Corollary~\ref{c:a-tick} implies that if we need to show existence of only one $X$-ticket in $T'$, 
then it is sufficient to prove that there is $x \in X'$ with $d_F(x) \geq 1$.
Part~(ii) of Corollary~\ref{c:a-tick} justifies claiming a ticket for {\em every} link $xz \in F$  
with $x \in X'$ and $z \notin T'$.
In the case $|M'| \geq 1$ we will use part~(i) only, since in this case 
existence of just one $X$-ticket in $T'$ makes $T'$ non-deficient.
In the case $|M'|=0$ we need to show existence of two $X$-tickets in $T'$,
but then we will use the Partial Solution Invariant (see Claim~\ref{c:X-ticket}) 
and part~(ii) of the corollary to show directly that no ticket is claimed for a link incident to a locked leaf.

\begin{figure}
\centering 
\epsfbox{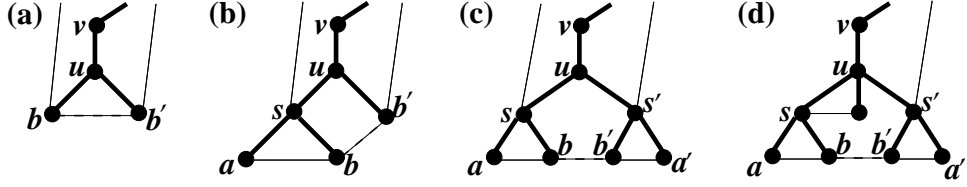}
\caption{Problematic trees. Links in $F'$ are shown by solid thin lines, and links in $M$ by dashed~lines.}
\label{f:bad-trees}
\end{figure}

We note that the reason we excluded twin links and locking links from $M$, and needed $N$-tickets,
was to avoid four specific ``{\em problematic trees}'' depicted in Fig.~\ref{f:bad-trees} 
of being deficient; assume that these trees have no non-leaf compound node. 
\begin{itemize}
\item[(a)]
If $M$ could include a twin link, $T',F'$ could be as in Fig.~\ref{f:bad-trees}(a)
(here $u$ may or may not be a stem). 
But, as we will show later, since $M$ has no twin link, then the Partial Solution Invariant 
implies that such $T'$ must contain a compound node, and thus cannot exist, see Claim~\ref{c:new-leaf}.
\item[(b)]
If $M$ could include locking links, $T',F'$ could be as in Fig.~\ref{f:bad-trees}(b),
where $s$ is a stem. 
But since $M$ has no locking link, such $T'$ cannot exist either.
\item[(c)]
Without the $N$-tickets, $T',F'$ could be as in Fig.~\ref{f:bad-trees}(c,d),
where $s,s'$ are stems.
Such $T'$ is not dangerous, since it has $2$ stems. 
Note that this $T'$ has an $N$-ticket. 
Thus $credit(T')=|U'|+3/2+1/2=|U'|+2$, while $|M'|+|U'|=|U'|+1$; 
consequently, such $T'$ is not deficient. 
\end{itemize}

%%%%%%%%%%%
\subsection{Properties of $T'$ under the Matching Invariant} \label{ss:M-inv}
%%%%%%%%%%%

Note that $M$ is a matching on the leaves of $T$, and under the Partial Solution Invariant, 
$M$ remains a matching on the leaves of $T/I$ and every leaf of $T/I$ matched by $M$
is an {\em original leaf}; this is so since we contract only $M$-compatible trees.
In particular, $d_F(b)=1$ for every $b \in L' \setminus U'$, by Claim~\ref{c:shm}
(note that $L' \setminus U'$ is the set of leaves of $T/I$ matched by $M$).
Recall that a $4$-leaf dangerous tree is ``reduced'' to a $3$-leaf dangerous tree
by contracting one twin link (see Fig.~\ref{f:def}).
However, this contraction is not $M$-compatible.
In order to use such a reduction, in this section we will temporarily 
replace the Partial Solution Invariant by the following weaker invariant,
where matched nodes of $T/I$ are leaves of $T/I$, but some of them may be compound nodes.

\begin{invariant} [(Matching Invariant)] \label{inv:M}
$M$ is a matching on the leaves of $T/I$ and $d_F(b)=1$ for every leaf $b$ of $T/I$ matched by $M$. 
\end{invariant}

The main purpose of this section is to prove the following properties of $T',F'$.

\begin{lemma} \label{l:3-lvs}
Suppose that the Credit Invariant and the Matching Invariant hold for $T$, $M$, and $I$,
and that $T/I$ has no link greedy contraction.
Then $|C'|=0$ (so $r \notin T'$) and either $|M'|=|S'|=0$, or 
$|M'|=1$, $|S'| \leq 2$, and $|L'| \leq |S'_1|+3$.
Furthermore, if $|L'|=3$ and $|M'|=1$, then one of the following holds: 
\begin{itemize}
\item[{\em (i)}]
$T',F'$ are as in Fig.~\ref{f:bad-trees}(b) (so $T'$ is a problematic tree with $3$ leaves), namely:
$T'$ has a unique stem $s$ with leaf descendants $a,b$, and another leaf $b'$,
such that $M'=\{bb'\}$ and $\{ab,sz,b'z'\} \subseteq F$ for some $z,z' \notin T'$.
\item[{\em (ii)}]
$T'$ is dangerous.
\end{itemize}
\end{lemma}

We emphasize again that in this section we use only the assumptions of Lemma~\ref{l:3-lvs},
and the Partial Solution Invariant is {\em not} assumed. 
(Under the Partial Solution Invariant, $b,b'$ are original leaves,
and case (i) is easily refuted by absence of locking links in $M$, 
as we already did when discussing the problematic tree in Fig.~\ref{f:bad-trees}(b).) 

\begin{claim} \label{c:CM} 
$|C'|=0$ and one of the following holds:
\begin{itemize}
\item[{\em (i)}]
$|M'|=|S'|=0$ and $tickets(T') \leq 1$.
\item[{\em (ii)}]
$|M'|=1$, $|S'| \leq 2$, and $tickets(T')=0$.
\end{itemize}
\end{claim}
\begin{proof}
Note that $coupons(T')=\frac{3}{2}|M'|+|U'|+|C'|$. Hence 
\begin{equation*}
credit(T')-(|M'|+|U'|)=|C'|+\frac{1}{2}|M'|+\frac{1}{2}tickets(T') \ .
\end{equation*}
From this we immediately get that either $|M'|=0$ and $tickets(T') \leq 1$,
or $|M'|=1$ and $tickets(T')=0$.
From the assumption that $T/I$ has no link greedy contraction we get
that every stem has a leaf descendant that is matched by $M$ 
(the link between two unmatched leaf descendants of a stem gives a link greedy contraction); 
this implies $|S'| \leq 2|M'|$.
Thus if $|M'|=0$ then $|S'|=0$, and if $|M'|=1$ then $|S'| \leq 2$.
\end{proof}

\begin{figure}
\centering 
\epsfbox{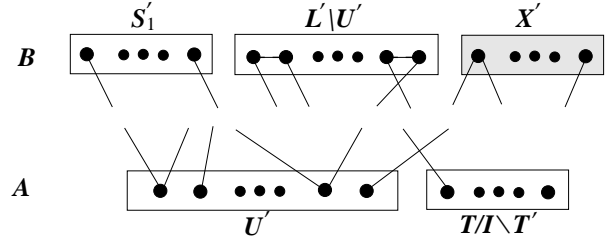}
   \caption{Illustration to the proof of Claim~\ref{c:count}.}
   \label{f:bipartite}
\end{figure}

\begin{claim} \label{c:count}
Let $A=U' \cup (T/I \setminus T')$. Then $|F(A,X')| \geq |U'|+1-2|M'|-|S'_1|$.
Consequently, if $tickets(T')=0$ then $|U'| \leq 2|M'|+|S'_1|-1$.
\end{claim}
\begin{proof}
Consider the set $F(A,T')$ of links in $F$ that have one end in $A$ and the other end in $T'$.
Note that:
\begin{itemize}
\item[$\bullet$]
There is no link between nodes in $U'$, since $T/I$ has no link greedy contraction.
\item[$\bullet$]
There is no link from $U'$ to $T/I \setminus T'$, since $T'$ is $U'$-closed.
\item[$\bullet$]
No link in $F$ is incident to a node in $S'\setminus S'_1$, 
by the definition of $S'_1$ and Claim~\ref{c:twins}.
\end{itemize}
Thus $F(A,T')=F(A,B)$ where $B=S'_1 \cup (L' \setminus U') \cup X'$ 
(see Fig.~\ref{f:bipartite}).
Now note that:
\begin{itemize}
\item[$\bullet$]
$|F(U',B)| \geq |U'|$, since $d_F(a) \geq 1$ for every $a \in U'$.
\item[$\bullet$]
$|F(T/I \setminus T',B)| \geq 1$, since some link in $F$ covers the parent edge of the root $v$ of $T'$.
\end{itemize} 
Thus 
$$|F(A,B)| \geq |U'|+1 \ .$$
On the other hand, $d_F(b)=1$ for any node $b \in B \setminus X = S'_1 \cup (L' \setminus U')$;
if $b \in S'_1$ then this is so by the definition of $S'_1$, and 
if $b \in L' \setminus U'$ is a leaf matched by $M$ then this is so by the Matching Invariant.
Consequently, 
$$|F(A,B \setminus X')| \leq |S'_1|+|L'|-|U'|=|S'_1|+2|M'| \ .$$
This implies 
$$
|F(A,X')|=|F(A,B)|-|F(A,B \setminus X')| \geq (|U'|+1)-(|S'_1|+2|M'|) 
$$
as claimed.

If $tickets(T')=0$ then $|F(A,X')|=0$; 
otherwise, there is $x \in X'$ with $d_F(x) \geq 1$, 
and by Corollary~\ref{c:a-tick}(i) $x$ has a ticket.
Thus $0=|F(A,X')| \geq |U'|+1-2|M'|-|S'_1|$, and the second statement follows.
\end{proof}

\begin{claim} \label{c:M=1}
If $|M'|=1$ then $|L'| \leq |S'_1|+3$. 
\end{claim}
\begin{proof}
By Claim~\ref{c:CM} $tickets(T')=0$ and thus 
by Claim~\ref{c:count} we must have $|U'| \leq 2|M'|+|S'_1|-1=|S'_1|+1$.
Consequently, $|L'|=|U'|+2|M'| \leq |S'_1|+1+2=|S'_1|+3$, as claimed.
\end{proof}

\begin{figure}
\centering 
\epsfbox{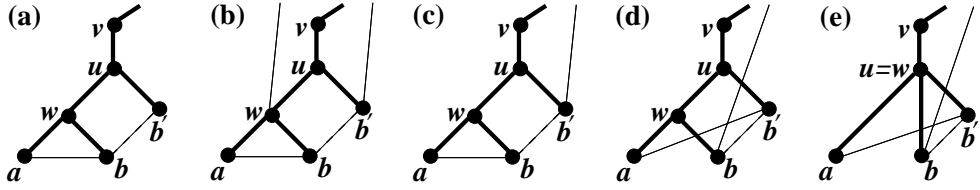}
\caption{Illustration to the proof of Claim~\ref{c:def3}. Links in $F$ are shown by solid thin lines. 
         In (b) $w$ is a stem.}
\label{f:def3}
\end{figure}

Now we prove the last part of Lemma~\ref{l:3-lvs}.

\begin{claim} \label{c:def3}
If $|L'|=3$ and $|M'|=1$, then $T',F'$ are as in Lemma~\ref{l:3-lvs}(i) or $T'$ is dangerous.
\end{claim}
\begin{proof}
Consider a general deficient $3$-leaf tree as in Fig.~\ref{f:def3}(a).
Since by Claim~\ref{c:CM} $C'=\emptyset$, 
the nodes that are not in $X'$ are $a,b,b'$, and possibly $w$ -- if $w$ is a stem.
Thus for any link $yz \in F$ with $y \in T'$, either $y=w$ is a stem, or $y \in \{a,b,b'\}$;
otherwise, $y$ has an $X$-ticket.
Now let $yz \in F$ be a link that covers the edge between $v$ and its parent, where $y \in T'$.
Note that $y \neq a$, since $T'$ is $a$-closed.
Thus either $y=w$ is a stem, or $y \in \{b,b'\}$.
Let us consider these two cases.

Suppose that there is a link $wz \in F'$ such that $z \notin T'$, so $y=w$ is a stem, see Fig.~\ref{f:def3}(b).
By Claim~\ref{c:twins}, $ab \in F$ and $d_F(w)=1$.
Now consider the (unique, by the Matching Invariant) link $b'z' \in F'$ incident to  $b'$.
We cannot have $z' \in \{a,b,w\}$ since $d_F(a)=d_F(b)=d_F(w)=1$
and we already have links $ab,wz$ in $F'$. 
Thus $z' \notin T'$ and we arrive at the problematic tree depicted in Fig.~\ref{f:def3}(b).

Suppose that $F'$ has no link $wz$ such that $z \notin T'$, so $y \in \{b,b'\}$.
Let $ax \in F'$ be a link that covers $a$.
To avoid an $X$-ticket at $x$, we must have that either $x=w$ is a stem or $x \in \{b,b'\}$.
The former case is not possible, by Claim~\ref{c:twins}. Thus $\{x,y\}=\{b,b'\}$, by the Matching Invariant.
Consequently, either $F'=\{ab,b'z\}$ (see Fig.~\ref{f:def3}(c)), or $F'=\{ab',bz\}$ (see Fig.~\ref{f:def3}(d)).
If $u \neq w$, then the former case $F'=\{ab,b'z\}$ (see Fig.~\ref{f:def3}(c)) is not possible;
this is since there must be a link in $F$ covering the edge between $w$ and its parent, 
but if $ab \in F$ then by Claim~\ref{c:shm} this link cannot be incident to one of $a,b$, 
and hence it gives a ticket. 
Thus $F'=\{ab',bz\}$, and we arrive at the case in Fig.~\ref{f:def3}(d), 
which is the $3$-leaf dangerous tree in Fig.~\ref{f:def}(a).
If $u=w$ (see Fig.~\ref{f:def3}(e) and Fig.~\ref{f:def}(b)), then 
there is no difference in the roles of $b,b'$, and we can have either $F'=\{ab,b'z\}$ or $F'=\{ab',bz\}$, 
obtaining in both cases a $3$-leaf dangerous tree.
\end{proof}

%%%%%%%%%%%
\subsection{Finishing the proof of Lemma~\ref{l:char}} \label{ss:char}
%%%%%%%%%%%

Recall that the Partial Solution Invariant implies that $M$ is a matching on the {\em original leaves} of $T/I$,
namely, that any node of $T/I$ matched by $M$ is an {\em original} leaf of $T$.
This and the fact that $M$ has no locking link, 
implies that the case in Lemma~\ref{l:3-lvs}(i) (see Fig.~\ref{f:bad-trees}(b)) cannot occur, 
as otherwise the leaf $a$ is locked by the link $bb' \in M$.
Thus Lemma~\ref{l:char} for $3$-leaf trees follows from Lemma~\ref{l:3-lvs}.

\begin{claim} \label{c:s'}
If $|M'|=1$ then $|S'_1| \leq 1$ and $|L'| \leq 4$, and if $|L'|=4$ then $|S'_1|=1$.  
\end{claim}
\begin{proof}
By Claim~\ref{c:CM}(ii) $|S'_1| \leq |S'| \leq 2$.
If $|S'_1|=2$, say $S'_1=\{s,s'\}$, then the link in $M$ connects
a twin of $s$ to a twin of $s'$.
This link has an $N$-ticket, by the definition of $S'_1$ and $N$. 
Thus $|S'_1| \leq 1$.
By Claim~\ref{c:M=1} $|L'| \leq |S'_1|+3$, and the statement follows.
\end{proof}

\begin{claim} \label{c:def4}
If $|M'|=1$ and $|L'|=4$ then $T'$ is dangerous. 
\end{claim}
\begin{proof}
By Claim~\ref{c:s'} $|S'_1|=1$; namely, 
$T'$ has a unique stem, say $s$, such that its twin link, say $f$, is in $F$.
Since $T/I$ has no greedy link contraction, $|M'|=1$, and since $M$ has no twin link,
exactly one of the twins of $s$ is matched by $M$.
Consider the tree $\tilde{T}=T/(I \cup \{f\})$ and its $3$-leaf subtree $\tilde{T}'$ 
obtained from $T/I$ and $T'$, respectively, by contracting $f$.
The contraction of $f$ creates a new leaf $b''$ that is now matched by $M$, and $b''$ is a leaf of $\tilde{T}'$.
This contraction is paid by the coupon of the unmatched twin of $s$,
and $b''$ does not need a coupon since it is matched by $M$;
hence the Credit Invariant holds for $\tilde{T}$, without overspending the credit.
By Claims~\ref{c:shm} and \ref{c:twins}, $\deg_F(b'')=1$, thus the Matching Invariant holds for $\tilde{T}$ and $M$.
Since $T/I$ has no link greedy contraction, $\tilde{T}$ has no link greedy contraction.
Consequently, the conditions of Lemma~\ref{l:3-lvs} hold for $\tilde{T}$, $M$, and $I \cup \{f\}$. 
Hence $\tilde{T}'$ must be a $3$-leaf tree as in Lemma~\ref{l:3-lvs},
as if $\tilde{T}'$ has a ticket, then so does $T'$.
Now note that $\tilde{T}'$ cannot be a problematic tree as in Fig.~\ref{f:bad-trees}(b), 
since then we will have $|S'_1|=2$, a case refuted in Claim~\ref{c:s'} by existence of an $N$-ticket.
Thus $\tilde{T}'$ is a $3$-leaf dangerous tree, as claimed.
\end{proof}

Note that if $|M'|=1$ and $|L'|=2$ then contracting the link in $M'$ creates a new leaf.
The following claim refutes this case by showing that in this case $T'$ must contain a compound node,
contradicting Claim~\ref{c:CM}.

\begin{figure}
\centering 
\epsfbox{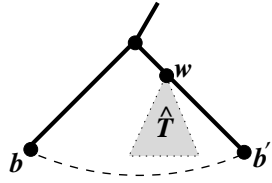}
   \caption{Illustration to the proof of Claim~\ref{c:new-leaf}.}
   \label{f:new-leaf}
\end{figure}

\begin{claim} \label{c:new-leaf}
If in $T/I$ contracting a link $bb' \in M$ creates a new leaf, then 
the path between $b$ and $b'$ in $T/I$ has an internal compound node.
\end{claim}
\begin{proof}
By the Partial Solution Invariant $b,b'$ are original leaves.
Note that in the original tree $T$ the contraction of $bb'$ does not create a new leaf, 
since $M$ has no twin link. 
This implies that in $T$, there is a subtree $\hat{T}$ of $T$ (see Fig.~\ref{f:new-leaf}) 
hanging out of a node $w$ on the path between $b$ and $b'$ in $T$. 
This subtree $\hat{T}$ is not present in $T/I$, hence it was contracted into a compound
node during the construction of our partial solution $I$. 
Thus $T/I$ contains a compound node $z$ that contains $\hat{T}$, and since $z$
contains a node $w$ that belongs to the path between $b$ and $b'$ in $T$,
the compound node of $T/I$ that contains $w$ belongs to the path between $b$ and $b'$ in $T/I$,
and it is distinct from $b,b'$ since both $b,b'$ are original leaves.
\end{proof}

To finish the proof of Lemma~\ref{l:char} it remains to refute the case $|M'|=0$.
In this case, the following claim together with Corollary~\ref{c:a-tick}(ii)
will enable us to claim $X$-tickets without worrying about links incident to locked leaves. 

\begin{figure}
\centering 
\epsfbox{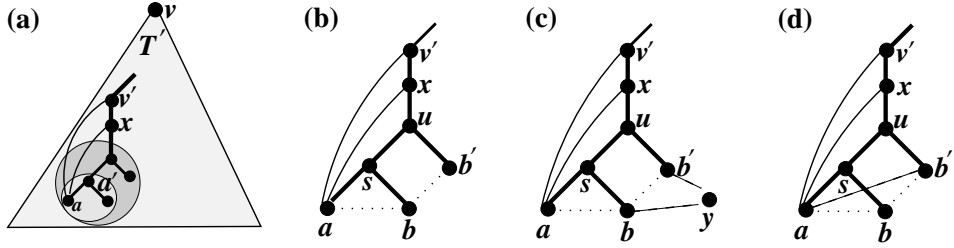}
   \caption{Illustration to the proof of Claim~\ref{c:X-ticket}.}
   \label{f:X-ticket}
\end{figure}

\begin{claim} \label{c:X-ticket}
Suppose that $|M'|=0$, and let $e=a'x \in F$ with $a' \in L'$ and $x \in X'$.
Then the original endnode $a$ of $e$ contained in $a'$ (possibly $a=a'$) is not a locked leaf, 
and thus $e$ contributes a ticket at $x$.
\end{claim}
\begin{proof}
Suppose to the contrary that $a$ is a locked leaf and let $T_{v'}$ be a locking tree of $a$ as 
in Definition~\ref{d:lock} (see Fig.~\ref{f:X-ticket}(a)). 
Consider 3 cases.

Case 1: No link in $M$ had an endnode in $T_{v'}$ (see Fig.~\ref{f:X-ticket}(b)). 
By the Partial Solution Invariant, we initially exhausted locking tree greedy contractions.
This  guarantees that $T_x$ lies in the same compound node $a'$ of $T'$, 
contradicting that $x \in X'$.

Case 2: There was a link in $M$ with exactly one endnode in $T_{v'}$ (see Fig.~\ref{f:X-ticket}(c)). 
This link is not incident to a node of $T'$, hence this link was contracted.
In particular, all nodes on the path between the endnodes of this link,
including $x$, lie in the same compound node of $T'$. This contradicts that $x \in X'$.

Case 3: No link in $M$ had exactly one endnode in $T_{v'}$ but there was a 
link in $M$ with both endnodes in $T_{v'}$ (see Fig.~\ref{f:X-ticket}(d)).
This link must be $ab'$, since initially $M$ had no twin link and no locking link. 
Since $|M'|=0$, the link $ab'$ does not appear in $T'$, hence it was contracted and both $a,b'$
lie in the same compound node of $T'$, which is $a'$.
Consider the first contraction when some node in $T_x$ entered a compound node. 
By the Partial Solution Invariant, this was either a greedy link contraction or a contraction of a semi-closed tree.
If it was a greedy link contraction, then it was between some node in $T_x$ and some node not in $T_x$,
since there is no greedy link contraction within $T_x$. But then also $x$ enters a compound node, 
contradicting that $x \in X \cap T'$.
If a contraction of a semi-closed tree occurred, then its root is a proper descendant of $x$
(since $x$ remains an original node), its leaf set is $\{a,b,b'\}$, and it is $b$-closed.
This implies that $ab'$ locks $b$, contradicting that $M$ has no locking links. 
\end{proof}

\begin{corollary}
If $|M'|=0$ then $tickets(T') \geq |L'|+1 \geq 2$, and thus $T'$ is not deficient.
\end{corollary}
\begin{proof}
Let $A=U' \cup (T/I \setminus T')$ be as in Claim~\ref{c:count}.
Every link in $F(X',A)$ contributes a ticket to $T'$;
for links in $F(X',T/I \setminus T')$ this is so by Corollary~\ref{c:a-tick}(ii), and 
for links in $F(X',U')$ this is so by Claim~\ref{c:X-ticket}.
By Claim~\ref{c:CM} $|S'|=0$, thus by Claim~\ref{c:count} 
$tickets(T') \geq |F(A,X')| \geq |U'|+1=|L'|+1 \geq 2$.
\end{proof}

The proof of Lemma~\ref{l:char} is now complete.

\paragraph*{Acknowledgment:}
We thank Andr\'{e} Linhares, Joseph Cheriyan, Zhihan Gao, Chaitanya Swamy,  
and three anonymous referees for many useful comments,
and Guy Even, Jon Feldman, and Samir Khuller for useful discussions.

% \bibliographystyle{acmtrans}
% \bibliography{tap}

\end{document}